\documentclass[journal]{IEEEtran}

\usepackage{amsmath,amssymb,mathtools,bm}
\usepackage{soul} 
\usepackage{array,booktabs}
\usepackage{makecell}
\usepackage{algorithmic}
\usepackage{algorithm}
\usepackage[hidelinks]{hyperref}
\usepackage{xcolor}
\newcommand{\revise}[1]{#1}
\newtheorem{lemma}{Lemma}

\newtheorem{theorem}{Theorem}

\title{A Necessary and Sufficient Condition for Exact Relaxation of Energy Storage Power Dispatch}

\author{Shan~Liu,~\IEEEmembership{Student Member,~IEEE},
Zhigang~Li,~\IEEEmembership{Senior Member,~IEEE},
and Ye~Guo,~\IEEEmembership{Senior Member,~IEEE}
\thanks{S. Liu and Z. Li are with the School of Science and Engineering,
The Chinese University of Hong Kong (Shenzhen), Shenzhen, China. Y. Guo is with the Department of Building Environment and Energy Engineering,
The Hong Kong Polytechnic University,
Hong Kong, China. (Correspondence to Zhigang Li. e-mail: lizg@cuhk.edu.cn)}}
\begin{document}

\bstctlcite{IEEEexample:BSTcontrol}

\maketitle

\begin{abstract}
Mixed-integer programming (MIP) precludes simultaneous charging and
discharging in power dispatch of energy storage system (ESS),
but is computationally intensive. Relaxation of this complementarity constraint improves
computational efficiency, but may incur an impractical control SCD strategy. To bridge this gap, this letter proposes a necessary and sufficient
condition for the exactness of such relaxation. We explicitly derive the relaxation gap and reveal that SCD does not necessarily imply a
distinguished optimal value from the MIP model. In such cases, a feasible solution without SCD can be readily recovered without compromising
the optimal value. Furthermore, a two-stage
relaxation-and-recovery algorithm is proposed to improve the computational efficiency of ESS power dispatch optimization.
\end{abstract}
\vspace{-1em}
\vspace{-1em}
\section{Introduction}
Complementarity constraints are conventionally included in energy storage system (ESS) models to exclude
impractical simultaneous charging and discharging (SCD) by formulating an mixed-integer programming (MIP) problem, which is acknowledged to
be computationally expensive \cite{lianganalytical2026}. Without introducing binary variables,
direct relaxation of the complementarity constraints avoids this dilemma effectively \cite{Pozo2022linear,Yildiran2023Robust}. However, such relaxation is not
necessarily exact, i.e., the relaxed model may yield an impractical SCD solution with a more optimistic optimal value than that of the
original problem. Moreover, the existing relevant literatures do not
theoretically distinguish whether the SCD of a relaxed solution implies a
nonzero relaxation gap or only a removable degeneracy. To bridge these gaps, we establish a necessary and sufficient condition for the exact relaxation of ESS power dispatch problems with complementarity constraints. \revise{We further analyze the impact of SCD on optimality and propose a two-stage relaxation-and-recovery algorithm to improve computational efficiency.}
\section{Problem Formulation and Analysis}
Consider the power dispatch problem of a single ESS over periods
$\mathcal{T}:=\{1,\ldots,T\}$
\begin{subequations}
\label{eq:original}
\begin{align}
Q:=\min\;
&
\sum_{t\in\mathcal{T}}
\left(
C_t^{\mathrm{CH}}p_t^{\mathrm{ch}}
+
C_t^{\mathrm{DC}}p_t^{\mathrm{dc}}
\right)
\label{eq:original_obj}\\
\mathrm{s.t.}\;
&
e_t=e_{t-1}
+\eta^{\mathrm{ch}}p_t^{\mathrm{ch}}\Delta t
-p_t^{\mathrm{dc}}\Delta t/\eta^{\mathrm{dc}},
\; t\in\mathcal{T},
\label{eq:soc_dyn}\\
&
\revise{0\le p_t^{\mathrm{ch}}\le u_tP_{\mathrm{ch}}^{\max},
\quad t\in\mathcal{T}},
\label{eq:ch_limit}\\
&
\revise{0\le p_t^{\mathrm{dc}}\le (1-u_t)P_{\mathrm{dc}}^{\max},
\quad t\in\mathcal{T}},
\label{eq:dc_limit}\\
&
u_t\in\{0,1\},
\quad t\in\mathcal{T},
\label{eq:binary}\\
&
\underline e_t\le e_t\le \overline e_t,
\quad t\in\mathcal{T},
\label{eq:soc_bound}\\
&
e_T=e_0.
\label{eq:terminal}
\end{align}
\end{subequations}
Here, $C_t^{\mathrm{CH}}$ and $C_t^{\mathrm{DC}}$ denote the
charging and discharging cost coefficients, respectively;
$\eta^{\mathrm{ch}},\eta^{\mathrm{dc}}\in(0,1]$ are the
corresponding efficiencies; \revise{
$P_{\mathrm{ch}}^{\max}>0$ and $P_{\mathrm{dc}}^{\max}>0$
are the charging and discharging power limits, respectively;
$\Delta t>0$ is the dispatch interval; and $e_0$ is the
initial state of charge (SoC). For notational compactness, define $\alpha:=P_{\mathrm{ch}}^{\max}/P_{\mathrm{dc}}^{\max}$.}

Let $Q^{\mathrm{R}}$ denote the optimal value of the 
relaxed problem where \eqref{eq:binary} is replaced with
$0\le u_t\le1$. Equivalently, constraints
\eqref{eq:ch_limit}--\eqref{eq:dc_limit} are recast to
\revise{\begin{equation}
p_t^{\mathrm{ch}}\geq0,\qquad
p_t^{\mathrm{dc}}\geq0,\qquad
p_t^{\mathrm{ch}}+\alpha p_t^{\mathrm{dc}}
\leq P_{\mathrm{ch}}^{\max}.
\label{eq:relaxed_power_region}
\end{equation}}
Indeed, multiplying the upper bound in \eqref{eq:dc_limit}
by $\alpha$ and adding it to \eqref{eq:ch_limit} yields the
last inequality in \eqref{eq:relaxed_power_region}. Conversely,
if \eqref{eq:relaxed_power_region} holds, setting
$u_t=p_t^{\mathrm{ch}}/P_{\mathrm{ch}}^{\max}\in[0,1]$ satisfies all constraints of the relaxed problem. Since the relaxed feasible set contains the original
feasible set, $Q^{\mathrm{R}}\le Q$; the relaxation is called
\emph{optimal-value exact} if $Q^{\mathrm{R}}=Q$. Define
$\eta:=\eta^{\mathrm{ch}}\eta^{\mathrm{dc}}$. We next characterize all
relaxed charging/discharging actions producing a fixed SoC transition.

\emph{Complementary action.} For a prescribed SoC transition, define
$\Delta e_t:=e_t-e_{t-1}$.
The complementary charging and discharging powers that realize this
transition are
\begin{equation}
p_t^{\mathrm{ch},0}:=
[\Delta e_t]^+/(\eta^{\mathrm{ch}}\Delta t),
\;
p_t^{\mathrm{dc},0}:=
\eta^{\mathrm{dc}}[-\Delta e_t]^+/\Delta t,
\label{eq:base_powers}
\end{equation}
respectively, where $[x]^+:=\max\{x,0\}$.

\begin{lemma}[Zero-SoC Cycling Decomposition]
\label{lem:zero_soc_decomposition}
Fix a feasible SoC transition $\Delta e_t$. Any relaxed pair
$(p_t^{\mathrm{ch}},p_t^{\mathrm{dc}})$ satisfying
\eqref{eq:soc_dyn} for this transition and
\eqref{eq:relaxed_power_region} can be uniquely represented as
\begin{align}
p_t^{\mathrm{ch}}
&=p_t^{\mathrm{ch},0}+z_t,
\label{eq:decomp_ch}\\
p_t^{\mathrm{dc}}
&=p_t^{\mathrm{dc},0}+\eta z_t,
\label{eq:decomp_dc}
\end{align}
for a scalar $z_t$ satisfying
\begin{equation}
0\le z_t\le \overline z_t(\Delta e_t),
\label{eq:z_bound}
\end{equation}
\end{lemma}
where
\revise{\begin{equation}
\revise{
\bar z_t(\Delta e_t):=
\bigl(P_{\mathrm{ch}}^{\max}
-p_t^{\mathrm{ch},0}
-\alpha p_t^{\mathrm{dc},0}\bigr)/(1+\alpha\eta)
} 
\label{eq:zbar}
\end{equation}}
\vspace{-1em}
\begin{IEEEproof}
If $\Delta e_t\ge0$, then $p_t^{\mathrm{dc},0}=0$ and
$p_t^{\mathrm{ch},0}=\Delta e_t/(\eta^{\mathrm{ch}}\Delta t)$. Since
$(p_t^{\mathrm{ch}},p_t^{\mathrm{dc}})$ produces the same transition,
\begin{equation*}
\Delta e_t=\eta^{\mathrm{ch}}p_t^{\mathrm{ch}}\Delta t
-\frac{p_t^{\mathrm{dc}}\Delta t}{\eta^{\mathrm{dc}}}
=
\eta^{\mathrm{ch}}p_t^{\mathrm{ch},0}\Delta t.
\end{equation*}
Rearranging gives
$p_t^{\mathrm{dc}}
=
\eta\left(p_t^{\mathrm{ch}}-p_t^{\mathrm{ch},0}\right).$
Thus
$z_t:=p_t^{\mathrm{ch}}-p_t^{\mathrm{ch},0}=p_t^{\mathrm{dc}}/\eta
\ge0$, and \eqref{eq:decomp_ch}--\eqref{eq:decomp_dc} follow.

If $\Delta e_t<0$, then $p_t^{\mathrm{ch},0}=0$. Equalities \eqref{eq:soc_dyn} and \eqref{eq:base_powers} give
$p_t^{\mathrm{dc}}
=
p_t^{\mathrm{dc},0}
+\eta p_t^{\mathrm{ch}}.$
Defining $z_t:=p_t^{\mathrm{ch}}\ge0$ yields
\eqref{eq:decomp_ch}--\eqref{eq:decomp_dc}.

The additional pair $(z_t,\eta z_t)$ produces zero SoC change since
$\eta^{\mathrm{ch}}z_t\Delta t-\eta z_t\Delta t/\eta^{\mathrm{dc}}=0$.
Substituting \eqref{eq:decomp_ch}--\eqref{eq:decomp_dc} into
\eqref{eq:relaxed_power_region} gives
\begin{equation*}
p_t^{\mathrm{ch},0}
+\alpha p_t^{\mathrm{dc},0}
+(1+\alpha\eta)z_t
\leq P_{\mathrm{ch}}^{\max},
\end{equation*}
which is equivalent to \eqref{eq:z_bound}--\eqref{eq:zbar}.
Uniqueness follows from
$z_t=p_t^{\mathrm{ch}}-p_t^{\mathrm{ch},0}
=(p_t^{\mathrm{dc}}-p_t^{\mathrm{dc},0})/\eta$.
\end{IEEEproof}

The scalar $z_t$ measures the zero-SoC cycling amount, i.e., the simultaneous charging and discharging component that leaves the SoC transition unchanged. $\overline z_t(\Delta e_t)$ is the maximum cycling amount permitted by the remaining power headroom for the prescribed SoC transition. Moreover, $\bar z_t(\Delta e_t)=0$ if and only if the
transition uses the full charging or discharging power capacity,
equivalently, \revise{$\Delta e_t=\eta^{\mathrm{ch}}
P_{\mathrm{ch}}^{\max}\Delta t$ or
$\Delta e_t=-P_{\mathrm{dc}}^{\max}\Delta t/
\eta^{\mathrm{dc}}$;} in this case, SCD is physically impossible
for that transition.
\vspace{-1em}
\section{Necessary and Sufficient Conditions for Relaxation Exactness}

This section first characterize stagewise relaxation exactness for a fixed SoC transition, and then lift the result to the full horizon. For a fixed SoC transition $\Delta e_t$, define the complementary stage
cost
$c_t(\Delta e_t)
:=
C_t^{\mathrm{CH}}p_t^{\mathrm{ch},0}
+
C_t^{\mathrm{DC}}p_t^{\mathrm{dc},0}.$
Using \eqref{eq:decomp_ch}--\eqref{eq:decomp_dc},
\begin{equation*}
C_t^{\mathrm{CH}}p_t^{\mathrm{ch}}
+
C_t^{\mathrm{DC}}p_t^{\mathrm{dc}}
=c_t(\Delta e_t)
+(C_t^{\mathrm{CH}}+\eta C_t^{\mathrm{DC}})z_t.
\end{equation*}
Define the marginal cycling cost
\begin{equation}
\kappa_t^0
:=
C_t^{\mathrm{CH}}
+
\eta C_t^{\mathrm{DC}}.
\label{eq:kappa_t}
\end{equation}
For the fixed transition $\Delta e_t$, the original problem forces
$z_t=0$, while the relaxation chooses $z_t$ over the interval
\eqref{eq:z_bound}. Hence the relaxed stage value is
\begin{align}
c_t^{\mathrm{R}}(\Delta e_t)
&:=
c_t(\Delta e_t)
+
\min_{0\le z_t\le \overline z_t(\Delta e_t)}
\kappa_t^0 z_t
\nonumber\\
&=
c_t(\Delta e_t)
-
[-\kappa_t^0]^+
\overline z_t(\Delta e_t).
\label{eq:relaxed_stage_closed}
\end{align}
Thus $c_t^{\mathrm{R}}(\Delta e_t)=c_t(\Delta e_t)$ if and only if
\begin{equation}
\kappa_t^0\ge 0
\quad\text{or}\quad
\overline z_t(\Delta e_t)=0.
\label{eq:stage_exact_condition}
\end{equation}

We next lift this stagewise characterization to the full horizon using
the feasible SoC trajectory set. The original and relaxed problems have
the same projection onto the SoC
trajectory space. Let
\begin{equation*}
\begin{aligned}
\mathcal{E}:=\bigl\{&
\bm e=(e_0,\ldots,e_T):\;
\underline e_t\le e_t\le \overline e_t,\quad t\in\mathcal{T},
\\
&
\revise{-P_{\mathrm{dc}}^{\max}\Delta t/\eta^{\mathrm{dc}}
\le e_t-e_{t-1}
\le \eta^{\mathrm{ch}}P_{\mathrm{ch}}^{\max}\Delta t,
\quad t\in\mathcal{T}},
\\
&
e_T=e_0
\bigr\},
\end{aligned}
\label{eq:E_set}
\end{equation*}
where the initial SoC $e_0$ is fixed. Indeed, every feasible solution of
either model has an SoC trajectory in $\mathcal{E}$. Conversely, for
any $\bm e\in\mathcal{E}$, define the complementary powers by
\eqref{eq:base_powers}. The transition bounds in $\mathcal E$ imply
$0\leq p_t^{\mathrm{ch},0}\leq
P_{\mathrm{ch}}^{\max}$ and
$0\leq p_t^{\mathrm{dc},0}\leq
P_{\mathrm{dc}}^{\max}$. Moreover,
$p_t^{\mathrm{ch},0}p_t^{\mathrm{dc},0}=0$, so one can choose a binary
mode $u_t$ to satisfy \eqref{eq:ch_limit}--\eqref{eq:binary}. Hence
$\bm e$ is implementable in the original problem.

Consequently, the two optimal values can be written as
\begin{align*}
Q
&=
\min_{\bm e\in\mathcal{E}}
\sum_{t\in\mathcal{T}}
c_t(e_t-e_{t-1}),
\\
Q^{\mathrm{R}}
&=
\min_{\bm e\in\mathcal{E}}
\sum_{t\in\mathcal{T}}
c_t^{\mathrm{R}}(e_t-e_{t-1}).
\end{align*}
\vspace{-1em}
\begin{theorem}[Necessary and Sufficient Exactness Condition] \label{thm:exactness}
The relaxation is optimal-value exact, i.e.,
$Q^{\mathrm{R}}=Q$, if and only if the relaxed problem admits an
optimal SoC trajectory $\bm e^*=(e_0^*,\ldots,e_T^*)\in\mathcal{E}$
such that
\begin{equation}
[-\kappa_t^0]^+
\overline z_t(e_t^*-e_{t-1}^*)=0,
\quad
\forall t\in\mathcal{T}.
\label{eq:full_exact_condition}
\end{equation}
Equivalently, for every $t\in\mathcal{T}$, $\kappa_t^0\ge0$ or
$\overline z_t(e_t^*-e_{t-1}^*)=0$.
\end{theorem}

\begin{IEEEproof}
\emph{1) Sufficiency:}
Suppose that the relaxed problem has an optimal SoC
trajectory $\bm e^*$ satisfying \eqref{eq:full_exact_condition}. By
\eqref{eq:relaxed_stage_closed},
$c_t^{\mathrm{R}}(e_t^*-e_{t-1}^*)=c_t(e_t^*-e_{t-1}^*)$ for all
$t\in\mathcal{T}$.
Thus
\begin{equation*}
Q
\le
\sum_{t\in\mathcal{T}}c_t(e_t^*-e_{t-1}^*)
=
\sum_{t\in\mathcal{T}}c_t^{\mathrm{R}}(e_t^*-e_{t-1}^*)
=
Q^{\mathrm{R}}.
\end{equation*}
Together with $Q^{\mathrm{R}}\le Q$, this gives $Q^{\mathrm{R}}=Q$.

\emph{2) Necessity:}
Suppose $Q^{\mathrm{R}}=Q$, and let
$\bm e^{\mathrm{O}}\in\mathcal{E}$ be an optimal SoC trajectory of the
original problem. Then
$\sum_{t\in\mathcal{T}}c_t(e_t^{\mathrm{O}}-e_{t-1}^{\mathrm{O}})=Q$.
The same trajectory is feasible for the relaxed projection, and
\eqref{eq:relaxed_stage_closed} implies
$c_t^{\mathrm{R}}(\Delta e_t)\le c_t(\Delta e_t)$ for every feasible
transition. Hence
\begin{equation*}
Q^{\mathrm{R}}
\le
\sum_{t\in\mathcal{T}}
c_t^{\mathrm{R}}(e_t^{\mathrm{O}}-e_{t-1}^{\mathrm{O}})
\le
\sum_{t\in\mathcal{T}}
c_t(e_t^{\mathrm{O}}-e_{t-1}^{\mathrm{O}})
=Q.
\end{equation*}
Since $Q^{\mathrm{R}}=Q$, both inequalities in
the preceding display are equalities. Therefore
$\bm e^{\mathrm{O}}$ is also optimal for the relaxed problem and
$\sum_{t\in\mathcal{T}}
[c_t(e_t^{\mathrm{O}}-e_{t-1}^{\mathrm{O}})
-c_t^{\mathrm{R}}(e_t^{\mathrm{O}}-e_{t-1}^{\mathrm{O}})]=0$.
Each summand is nonnegative, because by
\eqref{eq:relaxed_stage_closed},
\begin{equation*}
c_t(e_t^{\mathrm{O}}-e_{t-1}^{\mathrm{O}})
-
c_t^{\mathrm{R}}(e_t^{\mathrm{O}}-e_{t-1}^{\mathrm{O}})
=
[-\kappa_t^0]^+
\overline z_t(e_t^{\mathrm{O}}-e_{t-1}^{\mathrm{O}})
\ge0.
\end{equation*}
Hence every summand is zero; namely,
$[-\kappa_t^0]^+
\overline z_t(e_t^{\mathrm{O}}-e_{t-1}^{\mathrm{O}})=0$ for all
$t\in\mathcal{T}$. Thus, the relaxed problem admits the optimal trajectory
$\bm e^{\mathrm{O}}$ satisfying \eqref{eq:full_exact_condition}.
\end{IEEEproof}
\vspace{-1em}
\section{Physical Interpretation and Algorithm}
This section interprets the necessary and sufficient exactness condition
and proposes a two-stage relaxation-and-recovery algorithm for solving Problem~\eqref{eq:original}.

For an optimal relaxed SoC trajectory
$\bm e^\star\in\mathcal{E}$, the optimal value of the relaxed full-horizon problem is expressed by
\begin{equation}
	Q^{\mathrm{R}}
	=
	\min_{\bm e\in\mathcal{E}}
	\sum_{t=1}^{T}
	\left[
	c_t(\Delta e_t)
	+
	\min_{0\le z_t\le \overline z_t(\Delta e_t)}
	\kappa_t^0 z_t
	\right].
	\label{eq:relaxed_full_horizon}
\end{equation}
Hence the only possible gap from the original model is the zero-SoC
cycling term. Here, $\kappa_t^0$ is the marginal cycling cost, while
$\overline z_t(\Delta e_t^\star)$ is the available physical headroom
for such cycling. By Theorem~\ref{thm:exactness}, exactness along
$\bm e^\star$ is equivalent to
\begin{equation}
	\kappa_t^0\ge 0
	\quad\text{or}\quad
	\overline z_t(\Delta e_t^\star)=0,
	\qquad \forall t\in\mathcal{T}.
	\label{eq:physical_exactness_condition}
\end{equation}
Table~\ref{tab:stagewise} summarizes the possible cases. When
$\kappa_t^0>0$, cycling is strictly penalized and the relaxed optimum
satisfies $z_t^\star=0$. When $\kappa_t^0=0$ and $\overline z_t>0$, an
SCD solution may occur since all $z_t\in[0,\overline z_t]$ have the
same cost. The SCD component can be removed by setting $z_t=0$ in
\eqref{eq:decomp_ch}--\eqref{eq:decomp_dc}, so the returned SCD
solution is only a degenerate representation of an exact optimum.

If $\kappa_t^0<0$, the relaxed stage problem favors the largest
feasible cycling amount, $z_t^\star=\overline z_t(\Delta e_t^\star)$.
If $\overline z_t(\Delta e_t^\star)=0$, the selected transition already
uses the full charging or discharging power capacity, and the
relaxation remains exact. Otherwise, $z_t^\star>0$ and the relaxed
solution exploits profitable SCD. After solving the relaxed problem,
$z_t^\star=\min\{p_t^{\mathrm{ch},\star},
p_t^{\mathrm{dc},\star}/\eta\}$; a positive value indicates that binary
enforcement is required.
\vspace{-1em}
\begin{table}[!t]
	\caption{Stagewise Behavior of the Relaxed ESS Action}
	\label{tab:stagewise}
	\centering
	\footnotesize
	\setlength{\tabcolsep}{2pt}
	\renewcommand{\arraystretch}{1.12}
	\begin{tabular}{@{}>{\centering\arraybackslash}p{0.08\columnwidth}
			>{\centering\arraybackslash}p{0.08\columnwidth}
			>{\centering\arraybackslash}p{0.25\columnwidth}
			>{\centering\arraybackslash}p{0.18\columnwidth}
			>{\centering\arraybackslash}p{0.25\columnwidth}@{}}
		\toprule
		\makecell{$\kappa_t^0$}
		& \makecell{$\overline z_t$}
		& \makecell{Relaxed optimum\\$z_t^*$}
		& \makecell{SCD\\possible?}
		& \makecell{Exact for this\\SoC transition?}\\
		\midrule
		$>0$
		& Any
		& $z_t^*=0$
		& No
		& Yes\\
		$=0$
		& $>0$
		& Any $z_t^*\in[0,\overline z_t]$
		& Yes, possible
		& Yes, value exact\\
		$=0$
		& $=0$
		& $z_t^*=0$
		& No
		& Yes\\
		$<0$
		& $=0$
		& $z_t^*=0$
		& No
		& Yes\\
		$<0$
		& $>0$
		& $z_t^*=\overline z_t>0$
		& Yes
		& No\\
		\bottomrule
	\end{tabular}
	\vspace{-1em}
\end{table}

These observations motivate Algorithm~\ref{alg:two_stage_ess}. The
fully relaxed \revise{problem \eqref{eq:original}} is first solved to computes the
cycling amount $z_t$ at each period. If $\kappa_t^0>0$, SCD is
excluded in any optimum; if $\kappa_t^0=0$, SCD does not change the
objective and can be removed by setting $z_t=0$ in
\eqref{eq:decomp_ch}--\eqref{eq:decomp_dc}. Therefore, if no period in
$\mathcal{T}^{-}:=\{t\in\mathcal{T}:\kappa_t^0<0\}$ exhibits cycling,
the relaxed solution can be recovered directly. Otherwise, the problem
is resolved as an MIP with binary modes imposed only on
$\mathcal{T}^{-}$, followed by the same recovery at zero-cost periods.
Thus, at most one LP and one partial MIP problem need to be solved, and the resulting schedule is complementary and optimal.
\vspace{-1em}
\begin{algorithm}[]
	\caption{Two-Stage Relaxation-and-Recovery Method}
	\label{alg:two_stage_ess}
	\small
	\begin{algorithmic}[1]
		\REQUIRE ESS parameters
		\ENSURE \revise{Optimal complementary schedule$\{e_t^{\star},p_t^{\mathrm{ch},\star},
    p_t^{\mathrm{dc},\star}\}_{t\in\mathcal T}$}
		
		\STATE Compute \(\eta=\eta^{\mathrm{ch}}\eta^{\mathrm{dc}}\)
		and \(\kappa_t^0=C_t^{\mathrm{CH}}+\eta C_t^{\mathrm{DC}}\).
		\STATE Set \(\mathcal{T}^{-}=\{t\in\mathcal{T}:\kappa_t^0<0\}\)
		and \(\mathcal{T}^{0}=\{t\in\mathcal{T}:\kappa_t^0=0\}\).
		
		\STATE \revise{Solve the fully relaxed problem~\eqref{eq:original} to obtain
        $\{e_t,p_t^{\mathrm{ch}},p_t^{\mathrm{dc}}\}_{t\in\mathcal T}$.}
		\STATE Compute \(z_t=\min\{p_t^{\mathrm{ch}},p_t^{\mathrm{dc}}/\eta\}\),
		\(\forall t\in\mathcal{T}\), and set
		\(\mathcal{V}=\{t\in\mathcal{T}^{-}:z_t>0\}\).
		
		\IF{\(\mathcal{V}\neq\varnothing\)}
		\STATE \revise{Re-solve \eqref{eq:original} with binary $u_t$ on
        $\mathcal T^-$ and continuous $u_t$ otherwise to update
         $\{e_t,p_t^{\mathrm{ch}},p_t^{\mathrm{dc}}\}_{t\in\mathcal T}$.
		\STATE Recompute \(z_t\), \(\forall t\in\mathcal{T}^{0}\).}
		\ENDIF
		
		\FOR{\(t\in\mathcal{T}^{0}\) with \(z_t>0\)}
		\STATE Recover the complementary powers by
		\eqref{eq:base_powers} with \(\Delta e_t=e_t-e_{t-1}\)
		(equivalently, set \(z_t=0\) in
		\eqref{eq:decomp_ch}--\eqref{eq:decomp_dc}).
		\ENDFOR
		
		\RETURN \revise{The optimal schedule$\{e_t^{\star},p_t^{\mathrm{ch},\star},
    p_t^{\mathrm{dc},\star}\}_{t\in\mathcal T}$}
\end{algorithmic}
\end{algorithm}
\vspace{-2em}
\section{Numerical Experiments}
We consider a five-period ESS power dispatch with parameters
$\Delta t=1$, $\eta^{\mathrm{ch}}=\eta^{\mathrm{dc}}=0.95$,
$P_{\mathrm{ch}}^{\max}=P_{\mathrm{dc}}^{\max}=6$, $\underline e_t=0$, $\overline e_t=20$, and
$e_0=e_T=10$. The cost vectors are
$C^{\mathrm{CH}}=[10,10,2.5,1,2.7075]$ and
$C^{\mathrm{DC}}=[-5,-5,-5,5,-3]$, which give
$\kappa^0=[5.4875,5.4875,-2.0125,5.5125,0]$. All LP and MIP models are
solved by Gurobi. Algorithm~\ref{alg:two_stage_ess} is applied as
follows. In Step 1, the relaxed LP returns SCD at periods $t=3$ and
$t=5$, as shown in Table~\ref{tab:case2_lp}. Consistent with the
conclusion in Table~\ref{tab:stagewise}, since $\kappa_3^0<0$ and
$z_3>0$, the relaxation exploits profitable SCD that attains the lower
objective value $-35.5919$. In contrast, when $\kappa_5^0=0$, the SCD
at $t=5$ does not change the objective value, as indicated by
\eqref{eq:relaxed_stage_closed}. In Step 2, we impose a binary mode only
at $t=3$ and keep the other periods relaxed. The resulting partial MIP,
reported in Table~\ref{tab:case2_partial_mip}, has objective
$-30.4900$, matching the full MIP objective obtained by enforcing binary
modes in all periods. Although SCD remains at $t=5$, it can be removed
by \eqref{eq:base_powers}. Table~\ref{tab:case2_recovered} shows the
recovered complementary solution, whose objective remains $-30.4900$.
The experiment verifies the correctness of the proposed theory and algorithm.
\vspace{-1em}
\begin{table}[!t]
	\caption{Step 1: Relaxed LP Representative Solution}
	\label{tab:case2_lp}
	\centering
	\footnotesize
	\setlength{\tabcolsep}{3pt}
	\renewcommand{\arraystretch}{1.08}
	\begin{tabular}{@{}ccccccccc@{}}
		\toprule
		$t$ & $C_t^{\rm CH}$ & $C_t^{\rm DC}$ & $\kappa_t^0$ & $p_t^{\rm ch}$ & $p_t^{\rm dc}$ & $e_t$ & $z_t$ & SCD\\
		\midrule
		1 & 10 & -5 & 5.4875 & 0 & 3.5000 & 6.3158 & 0 & No\\
		2 & 10 & -5 & 5.4875 & 0 & 6 & 0 & 0 & No\\
		3 & 2.5000 & -5 & -2.0125 & 3.1537 & 2.8463 & 0 & 3.1537 & Yes\\
		4 & 1 & 5 & 5.5125 & 6 & 0 & 5.7000 & 0 & No\\
		5 & 2.7075 & -3 & 0 & 5.3009 & 0.6991 & 10 & 0.7746 & Yes\\
		\bottomrule
		\multicolumn{9}{@{}p{0.98\columnwidth}@{}}{\footnotesize
		$^{*}$ Units: costs and $\kappa_t^0$ in \$/MWh; powers and
		$z_t$ in MW; $e_t$ in MWh.}
	\end{tabular}
	\vspace{-1em}
\end{table}
\begin{table}[!t]
	\caption{Step 2: Partial MIP Solution with Binary Mode at $t=3$}
	\label{tab:case2_partial_mip}
	\centering
	\footnotesize
	\setlength{\tabcolsep}{3pt}
	\renewcommand{\arraystretch}{1.08}
	\begin{tabular}{@{}cccccccccc@{}}
		\toprule
		$t$ & $C_t^{\rm CH}$ & $C_t^{\rm DC}$ & $\kappa_t^0$ & $u_t$ & $p_t^{\rm ch}$ & $p_t^{\rm dc}$ & $e_t$ & $z_t$ & SCD\\
		\midrule
		1 & 10 & -5 & 5.4875 & -- & 0 & 4.7500 & 5 & 0 & No\\
		2 & 10 & -5 & 5.4875 & -- & 0 & 4.7500 & 0 & 0 & No\\
		3 & 2.5000 & -5 & -2.0125 & 1 & 6 & 0 & 5.7000 & 0 & No\\
		4 & 1 & 5 & 5.5125 & -- & 6 & 0 & 11.4000 & 0 & No\\
		5 & 2.7075 & -3 & 0 & -- & 2.4547 & 3.5453 & 10 & 2.4547 & Yes\\
		\bottomrule
	\end{tabular}
	\vspace{-1em}
\end{table}
\begin{table}[!t]
	\caption{Step 3: Recovered Complementary Solution}
	\label{tab:case2_recovered}
	\centering
	\footnotesize
	\setlength{\tabcolsep}{3pt}
	\renewcommand{\arraystretch}{1.08}
	\begin{tabular}{@{}ccccccccc@{}}
		\toprule
		$t$ & $C_t^{\rm CH}$ & $C_t^{\rm DC}$ & $\kappa_t^0$ & $p_t^{\rm ch}$ & $p_t^{\rm dc}$ & $e_t$ & $z_t$ & SCD\\
		\midrule
		1 & 10 & -5 & 5.4875 & 0 & 4.7500 & 5 & 0 & No\\
		2 & 10 & -5 & 5.4875 & 0 & 4.7500 & 0 & 0 & No\\
		3 & 2.5000 & -5 & -2.0125 & 6 & 0 & 5.7000 & 0 & No\\
		4 & 1 & 5 & 5.5125 & 6 & 0 & 11.4000 & 0 & No\\
		5 & 2.7075 & -3 & 0 & 0 & 1.3300 & 10 & 0 & No\\
		\bottomrule
\end{tabular}
	\vspace{-1em}
\end{table}
\section{Conclusion}

This letter provides a necessary and sufficient condition for
the relaxation exactness of a complementarity-constrained ESS power dispatch problem.
It shows that the relaxation gap is governed by the
marginal cycling cost $\kappa_t^0$ and the available cycling headroom
$\overline z_t$. In particular, the stagewise gap is
$[-\kappa_t^0]^+\overline z_t$, so SCD creates a true gap only when it
is both profitable and physically feasible. Therefore, the presence of
SCD in a relaxed solution does not necessarily mean that the relaxation
is inexact. When $\kappa_t^0=0$, SCD is only a degenerate zero-cost
cycling component, and setting $z_t=0$ recovers a complementary optimal
schedule with the same SoC trajectory and objective value. \revise{The proposed condition can guide solver design by identifying
unnecessary ESS mode binaries and enabling recovery of feasible solutions,
thereby improving computational efficiency. Future work will
integrate this theory into network-constrained power system optimization problems.}
\vspace{-1em}
\bibliographystyle{IEEEtran}
\bibliography{references}

\end{document}